\documentclass[preprint,12pt,authoryear]{elsarticle}

\usepackage[OT1]{fontenc}
\usepackage{amsthm,amsmath}
\usepackage{natbib}
\usepackage[colorlinks,citecolor=blue,urlcolor=blue]{hyperref}
\usepackage{amsfonts}
\usepackage{graphicx}
\usepackage{float}
\usepackage{amssymb}
\usepackage{bbm}
\usepackage{lipsum}
\makeatletter
\def\ps@pprintTitle{%
	\let\@oddhead\@empty
	\let\@evenhead\@empty
	\def\@oddfoot{}%
	\let\@evenfoot\@oddfoot}
\makeatother

\newtheorem{prop}{\textbf{Proposition}}

\def\argmax{\operatorname{argmax} \displaylimits}

\def\CGM{\mathrm{CGM}}
\def\conv{\mathrm{conv}}

\def\d{\delta}

\def\e{\epsilon}

\def\E{\mathbb{E}}
\def\EM{\mathrm{EM}}
\def\F{\mathcal{F}}
\def\FS{\mathrm{FS}}
\def\G{\mathcal{G}}

\def\ISIG{\mathrm{ISIG}}
\def\L{\Lambda}
\def\l{\lambda}

\def\TSE{\mathrm{TSE}}

\def\R{\mathbb{R}}

\def\s{\sigma}

\def\th{\theta}
\def\Th{\Theta}
\def\T{\top}

\newcommand{\bel}{\begin{eqnarray}\label}
\newcommand{\eel}{\end{eqnarray}}
\newcommand{\bes}{\begin{eqnarray*}}
\newcommand{\ees}{\end{eqnarray*}}
\newcommand{\bei}{\begin{itemize}}
\newcommand{\beiftnt}{\begin{itemize}\footnotesize}
\newcommand{\eei}{\end{itemize}}

\def\benu{\begin{enumerate}}
\def\eenu{\end{enumerate}}

\def\argmax{\mathop{\rm arg\, max}}

\def\real{{\mathbb{R}}}
\def\R{{\real}}

\def\E{{\mathbb{E}}}
\def\G{{\mathbb{G}}}
\def\P{{\mathbb{P}}}

\def\complex{\mathop{{\rm I}\kern-.58em\hbox{\rm C}}\nolimits}

\def\muhat{\widehat{\mu}}

\journal{Computational Statistics \& Data Analysis}

\begin{document}

\begin{frontmatter}



\title{Approximate Nonparametric Maximum Likelihood for Mixture Models: A Convex Optimization Approach to Fitting Arbitrary Multivariate Mixing Distributions}

\author{Long Feng\fnref{label1}}
\author{Lee H. Dicker}
\fntext[label1]{Correspondence to: 110 Frelinghuysen Rd, Piscataway, NJ 08854\\
E-mail addresses: long.feng@rutgers.edu (Long Feng), ldicker@stat.rutgers.edu (Lee H. Dicker)} 
\address{Department of Statistics and Biostatistics, Rutgers University}

\begin{abstract}
Nonparametric maximum likelihood (NPML) for mixture models is a
technique for estimating mixing distributions that has a long and
rich history in statistics going back to the 1950s, and is closely
related to empirical Bayes methods.  Historically, NPML-based
methods have been considered to be relatively impractical because of
computational and theoretical obstacles.  However, recent work
focusing on approximate NPML methods suggests that these methods may
have great promise for a variety of modern applications.  Building
on this recent work, a class of flexible, scalable, and
easy to implement approximate NPML methods is studied for problems with
multivariate mixing distributions.  Concrete guidance on
implementing these methods is provided, with theoretical and empirical support;
topics covered include identifying the support set of the mixing
distribution, and comparing algorithms (across a variety of metrics)
for solving the simple convex optimization problem at the core of
the approximate NPML problem.  Additionally, 
three diverse real data applications are studied to illustrate the methods' performance: (i) A baseball
data analysis (a classical example for empirical Bayes methods), (ii) high-dimensional
microarray classification, and (iii) online prediction of
blood-glucose density for diabetes patients.  Among other things,
the empirical results demonstrate the relative effectiveness of using multivariate (as opposed to univariate) mixing distributions for NPML-based approaches.

\end{abstract}


\end{frontmatter}


\section{Introduction}\label{sec:intro}

Consider a setting where we have iid observations from a mixture
model.  More specifically, let $G_0$ be a probability distribution on
$\mathcal{T} \subseteq \R^d$ and let $\{
F_0(\cdot\mid\th)\}_{\th \in \mathcal{T}}$ be a family of probability
distributions on $\R^n$ indexed by
the parameter $\th \in \mathcal{T}$.  Throughout the paper, we assume
that $\mathcal{T}$ is closed and convex.  Assume that $X_1,...,X_p
\in \R^n$ are observed iid random variables and that $\Th_1,...,\Th_p \in
\R^d$ are corresponding iid latent variables, which satisfy
\begin{equation}\label{MM}
X_j\mid\Th_j \sim F_0(\cdot\mid\Th_j) \mbox{ and } \Th_j \sim G_0.  
\end{equation}
In \eqref{MM}, it may be the case that $F_0(\cdot\mid\th)$ and 
$G_0$ are both known (pre-specified) distributions; more frequently,
this is not the case.  In this paper, we will study problems where the mixing distribution $G_0$ is
unknown, but will assume
$F_0(\cdot\mid\th)$ is known throughout.  Problems like this arise in
applications throughout statistics, and various solutions have been proposed.
The distribution $G_0$ can be modeled parametrically, which leads to
hierarchical modeling and parametric empirical Bayes methods \cite[e.g.][]{efron2010large}.  Another
approach is to model $G_0$ as a discrete distribution supported on
finitely- or infinitely-many points; this leads to the study of
finite mixture models or nonparametric Bayes, respectively \citep{mclachlan2004finite,ferguson1973bayesian}.   This
paper focuses on another method for estimating $G_0$: Nonparametric maximum
likelihood.

Nonparametric maximum likelihood (NPML) methods for mixture
models --- and closely related empirical Bayes methods --- have been studied in statistics since the 1950s
\citep{robbins1950generalization, kiefer1956consistency,robbins1956empirical}.  They make
virtually no assumptions on the mixing distribution $G_0$ and provide
an elegant approach to problems like \eqref{MM}.   The general strategy is to first find the nonparametric maximum likelihood estimator (NPMLE) for $G_0$, denoted by $\hat{G}$, then perform inference via empirical Bayes \citep{robbins1956empirical,efron2010large}; that is, inference in \eqref{MM} is conducted via the posterior distribution $\Th_j\mid X_j$, under the assumption $G_0 = \hat{G}$.
Research into NPMLEs for mixture models has included
work on algorithms for computing NPMLEs and theoretical work on their
statistical properties \citep[e.g.][]{laird1978nonparametric,bohning1992computer,lindsay1995mixture,ghosal2001entropies,jiang2009general}.  However,
implementing and analyzing NPMLEs for mixture models has
historically been considered very challenging  \citep[e.g. p. 571
of][]{dasgupta2008asymptotic,donoho2013achieving}.  In this paper, we
study a computationally convenient approach involving approximate NPMLEs, which sidesteps many
of these difficulties and is shown to be effective in a variety of applications.

Our approach is largely motivated by recent work initiated by \cite{koenker2014convex}\footnote{Koenker \& Mizera's work was itself partially inspired by relatively recent theoretical work on NPMLEs by \cite{jiang2009general}.} and further pursued by others, including \cite{gu2015problem, gu2017empirical, gu2017unobserved} and \cite{dicker2016nonparametric}.  Koenker \& Mizera studied convex approximations to
NPMLEs for mixture
models in relatively large-scale problems, with up to 100,000s of observations.  In
\cite{koenker2014convex}, they showed that for the Gaussian location model, where $X_j =
\Th_j + Z_j \in \R$ and $\Th_j \sim G_0$, $Z_j \sim N(0,1)$ are
independent, a good approximation to the NPMLE for $G_0$ can be
accurately and
rapidly computed using generic interior point methods.

\cite{koenker2014convex}'s focus on convexity
and scalability is one of the key concepts for this paper.  Here, we show how a simple convex approximation to the NPMLE can
be used effectively in a broad range of problems with nonparametric mixture models; including problems
involving (i) multivariate mixing distributions, (ii) discrete data,
(iii) high-dimensional classification, and (iv) state-space models. Backed by new theoretical and empirical results, we provide concrete
guidance for efficiently and reliably computing approximate multivariate
NPMLEs.  Our main theoretical result (Proposition \ref{prop:conv}) suggests a simple procedure for finding the support set of the estimated mixing distribution.  Many of our empirical results highlight the benefits of using multivariate mixing distributions with correlated components (Sections \ref{sec:numMean}, \ref{sec:bball}, and \ref{sec:class}), as opposed to univariate mixing distributions, which have been the primary focus of previous research in this area (notable exceptions include theoretical work on the Gaussian location-scale model in \cite{ghosal2001entropies} and applications in \cite{gu2017empirical,gu2017unobserved} involving estimation problems with Gaussian models).  In Sections \ref{sec:bball}--\ref{sec:cgm}, we illustrate the performance of the methods described here in real-data
applications involving baseball, cancer microarray data, and online
blood-glucose monitoring for diabetes patients.  In comparison with other recent work on NPMLEs for mixture models, this paper distinguishes itself from \cite{gu2015problem} in that it focuses on more practical aspects of fitting general multivariate NPMLEs.  Additionally, in this paper we consider a substantially broader swath of applications than \cite{gu2017empirical,gu2017unobserved} and \cite{dicker2016nonparametric}, where the focus is estimation in Gaussian models and classification with a univariate NPMLE, respectively, and show that the same fundamental ideas may be effectively applied in all of these settings.

\section{NPMLEs for mixture models via convex optimization}
\label{sec:NPMLE}

\subsection{NPMLEs}
\label{sec:NPMLEsub}

Let $\G_{\mathcal{T}}$ denote the class of all probability distributions on
$\mathcal{T} \subseteq \R^d$ and suppose that $f_0(\cdot\mid\th)$ is the
probability density corresponding to $F_0(\cdot\mid\th)$ (with respect to
some given base measure).  For $G \in \G_{\mathcal{T}}$, the (negative) log-likelihood given the
data $X_1,...,X_p$ is 
\[
\ell(G) =  -\frac{1}{p}\sum_{j = 1}^p \log \left\{\int_{\mathcal{T}}
f_0(X_j\mid\th) \ dG(\th)\right\}.
\]
The Kiefer-Wolfowitz NPMLE for $G_0$ \citep{kiefer1956consistency},
denoted $\hat{G}$, solves the optimization problem
\begin{equation}\label{KW}
\min_{G \in \G_{\mathcal{T}}} \ell(G);  
\end{equation}
in other words, $\ell(\hat{G}) = \min_{G \in \G_{\mathcal{T}}}
\ell(G)$.  

Solving
\eqref{KW} and studying properties of $\hat{G}$ forms the basis for
basically all of the existing research into NPMLEs for mixture models
(including this paper).  Two important observations have had significant but somewhat
countervailing effects on this research: 
\begin{itemize}
	\item[(i)] The optimization problem
	\eqref{KW} is convex;
	\item[(ii)] If $f_0(X_j|\th)$ and $\mathcal{T}$ satisfy certain
	(relatively weak) regularity conditions, then $\hat{G}$ exists and may be chosen
	so that it is a discrete measure supported on at most $p$
	points.  
\end{itemize}
The first observation above is obvious; the second 
summarizes Theorems 18--21 of \cite{lindsay1995mixture}. Among the
more significant regularity conditions mentioned in (ii) 
is that the set $\{f_0(X_j|\th)\}_{\th \in \mathcal{T}}$ should be bounded for
each $j = 1,...,p$.  

Observation (i) leads to KKT-like conditions that characterize
$\hat{G}$ in terms of the gradient of $\ell$ and can be used to develop
algorithms for solving
\eqref{KW} \citep[e.g.][]{lindsay1995mixture}.  While this approach
is somewhat appealing, \eqref{KW} is typically an
infinite-dimensional optimization problem (whenever $\mathcal{T}$ is
infinite).  Hence, there are infinitely
many KKT conditions to check, which is generally impossible in
practice.

On the other hand, observation (ii) reduces \eqref{KW} to a
finite-dimensional optimization problem.  Indeed, (ii)
implies that $\hat{G}$ can be found by restricting attention in
\eqref{KW} to $G \in \G_p$, where $\G_p$ is the set of discrete probability measures
supported on at most
$p$ points in $\mathcal{T}$.  Thus, finding $\hat{G}$ is reduced to
fitting a finite mixture model with at most $p$ components.  This is
usually done with the EM-algorithm \citep{laird1978nonparametric}, where
in practice one may restrict to $G \in \G_q$ for some $q < p$. More recently, \citet{wang2015nonparametric} proposed to estimate $G$ in a multivariate setting using a gradient based method.
However, while (ii) reduces \eqref{KW} to a
finite-dimensional problem, we have lost convexity: 
\begin{equation}\label{ncvx}
\min_{G \in \G_q}
\ell(G)
\end{equation}
is not a convex problem because $\G_q$ is nonconvex.  When $q$
is large (and recall that the theory suggests we should take $q=p$),
well-known issues related to nonconvexity and finite mixture models
become a significant obstacle \citep{mclachlan2004finite}.

\subsection{A simple finite-dimensional convex approximation}

In this paper, we take a very simple approach to (approximately)
solving \eqref{KW}, which maintains convexity and
immediately reduces \eqref{KW} to a finite-dimensional problem. Consider a pre-specified finite grid $\L \subseteq \mathcal{T}$.  We study estimators $\hat{G}_{\L}$,
which solve
\begin{equation}\label{KM}
\min_{G \in \G_{\L}} \ell(G).
\end{equation}
The key difference between \eqref{ncvx} and \eqref{KM} is that
$\G_{\L}$, and hence \eqref{KM}, is convex, while $\G_q$ is
nonconvex.  Additionally, \eqref{KM} is a finite-dimensional
optimization problem,
because $\L$ is finite.  

To derive a more convenient formulation of \eqref{KM}, suppose that 
\begin{equation}\label{Lfinite}
\L = \{t_1,...,t_q\} \subseteq \mathcal{T}
\end{equation}
and define the simplex $\Delta^{q-1} = \{w = (w_1,...,w_q) \in\R^q; \ w_k \geq 0, \ w_1 +
\cdots + w_q = 1\}$.   Additionally, let $\delta_t$ denote a point
mass at $t \in \R^d$.  
Then there is a correspondence between $G = \sum_{k = 1}^q w_k\d_{t_k} \in \mathcal{G}_{\L}$ and points
$w = (w_1,...,w_q) \in \Delta^{q-1}$.  It follows that \eqref{KM} is equivalent to the optimization problem over the
simplex, 
\begin{equation}\label{KMsimp}
\min_{w \in \Delta^{q-1}} -\frac{1}{p} \sum_{j = 1}^p
\log\left\{\sum_{k = 1}^q f_0(X_j|t_k)w_k\right\}.
\end{equation} 

Researchers studying NPMLEs have previously considered estimators like $\hat{G}_{\L}$, which solve
\eqref{KM}--\eqref{KMsimp}.  However,
most have focused on relatively simple models with univariate mixing distributions $G_0$
\citep{bohning1992computer,jiang2009general,koenker2014convex}.  In very recent work, \cite{gu2017empirical, gu2017unobserved}  have considered multivariate NPMLEs for estimation problems involving Gaussian models --- by contrast, our aim is to formulate strategies for solving and implementing the general problem, as specified in \eqref{KW}, \eqref{KM}, and \eqref{KMsimp}.

\section{Choosing $\L$}
\label{sec:L}

The approximate NPMLE $\hat{G}_{\L}$ is the estimator we use
throughout the rest of the paper.  One remaining question is: How
should $\L$ be chosen?  Our perspective is that $\hat{G}_{\L}$ is an
approximation to $\hat{G}$ and its performance characteristics
are inherited from $\hat{G}$.  In general, $\hat{G}_{\L} \neq \hat{G}$.  However, as
one selects larger and larger finite grids 
$\L \subseteq \mathcal{T}$, which are more and more dense in
$\mathcal{T}$, evidently $\hat{G}_{\L} \to
\hat{G}$.  Thus, heuristically, as long as the grid $\L$ is ``dense enough'' in $\mathcal{T}$,
$\hat{G}_{\L}$ should perform similarly to $\hat{G}$.  

If $\mathcal{T}$ is compact, then any regular grid $\L \subseteq
\mathcal{T}$ is finite and implementing \eqref{KM} is
straightforward (specific implementations are discussed in Section \ref{sec:implementation}).  Thus, for compact $\mathcal{T}$, one can choose $\L$ to be a
regular grid with as many points as are computationally feasible.  For general $\mathcal{T}$, we propose a
two-step approach to choosing $\L$:  (i) Find a compact convex subset $\mathcal{T}_0
\subseteq \mathcal{T}$ so that \eqref{KW} is equivalent (or
approximately equivalent) to 
\begin{equation}\label{KWcpct}
\inf_{G
	\in \mathcal{G}_{\mathcal{T}_0}} \ell(G);
\end{equation}
(ii) choose $\L \subseteq \mathcal{T}_0 \subseteq \mathcal{T}$ to be a
regular grid with $q$
points, for some sufficiently large $q$.  Empirical results
seem to be fairly insensitive to the choice of $q$.  In Sections \ref{sec:bball}--\ref{sec:cgm}, we choose $q = 30^d$ for models with $d=2,3$ dimensional
mixing distributions $G$.  For some simple models with univariate $G$
($d=1$), theoretical results suggest that if $q = \sqrt{p}$, then $\hat{G}_{\L}$ is
statistically indistinguishable from $\hat{G}$
\citep{dicker2016nonparametric}.

Now we form strategies to construct $\Lambda$ for multivariate mixing distribution. First note that \citet{lindsay1981preoperties,lindsay1995mixture} has proposed the construction of $\Lambda$ in one dimensional cases. Our approach may be viewed as a generalization of their work. For each $j =
1,\ldots,p$, define 
\[
\hat{\th}_j =\hat{\th}(X_j) = \argmax_{\th \in \mathcal{T}}
f_0(X_j\mid \th)
\]
to be the maximum likelihood estimator (MLE) for $\Th_j$, given the data
$X_j \in \R^n$.  The following proposition implies that \eqref{KW} and
\eqref{KWcpct} are equivalent when the likelihoods $f_0(X_j\mid \th)$
are from a class of elliptical unimodal distributions, and $\mathcal{T}_0 =
\conv(\hat{\th}_1,\ldots,\hat{\th}_p)$ is the convex hull of
$\hat{\th}_1,\ldots, \hat{\th}_p$.   This result enables us to employ the strategy described above
for choosing $\L$ and finding $\hat{G}_{\L}$; specifically, we take
$\L$ to be a regular grid contained in the compact convex set
$\conv(\hat{\th}_1,\ldots,\hat{\th}_p)$.    

\begin{prop}\label{prop:conv}
	Suppose that $f_0$ has the form
	\begin{equation}\label{unimodal}
	f_0(X_j \mid \th) = h\big\{(\hat{\th}_j -
	\th)^{\T}\Sigma^{-1}(\hat{\th}_j - \th)\big\} u(X_j), 
	\end{equation}
	where $h: [0,\infty) \to [0,\infty)$ is a decreasing
	function, $\Sigma$ is a $p \times p$ positive definite matrix, and $u:\R^n
	\to \R$ is some other function that does not depend on $\th$.  Let $\mathcal{T}_0 =
	\conv(\hat{\th}_1,\ldots,\hat{\th}_p)$.  Then
	$\ell(\hat{G}) = \inf_{G \in \G_{\mathcal{T}_0}}
	\ell(G)$.
\end{prop}

\begin{proof}
	
	Assume that $G = \sum_{k = 1}^q w_k \d_{t_k}$, where $t_1,\ldots,t_q
	\in \mathcal{T}$ and $w_1,\ldots,w_q > 0$.  Further assume that $t_q
	\notin \mathcal{T}_0= \conv(\hat{\th}_1,\ldots,\hat{\th}_p)$.  We show
	that there is another probability distribution $\tilde{G} = \sum_{k =
		1}^{q-1}w_k\d_{t_k} + w_q\d_{\tilde{t}_q}$, with $\tilde{t}_q \in
	\mathcal{T}_0$, satisfying $\ell(\tilde{G}) \leq \ell(G)$.  This
	suffices to prove the proposition.  
	
	Let $\tilde{t}_q$ be the projection of $t_q$ onto $\mathcal{T}_0$ with
	respect to the inner product $(s,t) \mapsto s^{\T}\Sigma^{-1}t$.  To
	prove that $\ell(\tilde{G}) \leq \ell(G)$, we show that $f_0(X_j\mid
	\tilde{t}_q) \geq f_0(X_j\mid t_q)$ for each $j = 1,\ldots,p$.  We have
	\begin{align*}
	(\hat{\th}_j - t_q)^{\T}\Sigma^{-1}(\hat{\th}_j - t_q) & =(\hat{\th}_j
	-
	\tilde{t}_q +
	\tilde{t}_q -
	t_q)^{\T}\Sigma^{-1}(\hat{\th}_j
	-
	\tilde{t}_q +
	\tilde{t}_q -
	t_q)  \\
	& = (\hat{\th}_j - \tilde{t}_q)^{\T}\Sigma^{-1}(\hat{\th}_j -
	\tilde{t}_q) + 2(\tilde{t}_q - t_q)^{\T}\Sigma^{-1}(\hat{\th}_j -
	\tilde{t}_q) \\
	& \qquad + (\tilde{t}_q - t_q)^{\T}\Sigma^{-1}(\tilde{t}_q -
	t_q) \\
	& \geq (\hat{\th}_j - \tilde{t}_q)^{\T}\Sigma^{-1}(\hat{\th}_j -
	\tilde{t}_q),
	\end{align*}
	where we have used the fact that $(\tilde{t}_q - t_q)^{\T}\Sigma^{-1}(\hat{\th}_j -
	\tilde{t}_q) = 0$, because $\tilde{t}_q$ is the projection of
	$t_q$ onto $\mathcal{T}_0$.  By \eqref{unimodal}, it follows that
	$f_0(X_j\mid \tilde{t}_q) \geq f_0(X_j\mid t_q)$, as was to be
	shown.  
\end{proof}

The
condition \eqref{unimodal} is rather restrictive, but we believe it
applies in a number of important problems.  The
fundamental example where \eqref{unimodal} holds is $X_j \mid
\Th_j \sim N(\Th_j,\Sigma)$;  in
this case $\hat{\th}_j = X_j$ and \eqref{unimodal} holds with $u$ being certain constant and $h(z) \propto e^{-z/2}$.
Condition \eqref{unimodal} also holds in elliptical models, 
where $\Th_j$ is the location parameter of $X_j\mid \Th_j$.  More
broadly,  if $X_j = (X_{1j},\ldots,X_{nj})^{\T}  \in \R^n$ may be viewed as a
vector of replicates $X_{ij}$, $i = 1,\ldots,n$, drawn from some common
distribution conditional on $\Th_j$, then standard results suggest that the MLEs $\hat{\th}_j$ may be
approximately Gaussian if $n$ is sufficiently large, and
\eqref{unimodal} may be approximately valid.  Specific applications where
a normal approximation argument for $\hat{\th}_j$ may imply that
\eqref{unimodal} is approximately valid include count data (similar to
Section
\ref{sec:bball}) and time series modeling (Section \ref{sec:cgm}).

\section{Connections with finite mixtures}

Finding $\hat{G}_{\L}$ is equivalent to fitting a
finite-mixture model, where the locations of the atoms for the mixing
measure have been
pre-specified (specifically, the atoms are taken to be the points in
$\L$).  Thus, the approach in this paper reduces computations for the relatively complex
nonparametric mixture model \eqref{MM} to a convex optimization problem that is substantially
simpler than fitting a standard finite mixture model (generally a
non-convex problem). 

An important distinction of nonparametric mixture
models is that they lack the built-in interpretability of the
components/atoms from finite mixture models, and are less suited for clustering applications.  On the other hand, taking the
nonparametric approach provides additional flexibility for modeling heterogeneity in applications where it is not clear that
there should be well-defined clusters.  Moreover, post hoc clustering and finite
mixture model methods could still be used after
fitting an NPMLE; this might be advisable if, for instance, $\hat{G}_{\L}$
has several clearly-defined modes.

\section{Implementation overview}
\label{sec:implementation}

A variety of well-known algorithms are available for solving
\eqref{KMsimp} and finding $\hat{G}_{\L}$.  We have experimented with several, including the EM-algorithm, interior point methods, and the Frank-Wolfe algorithm.  This section contains a brief overview of how we have implemented these algorithms; numerical results comparing the algorithms are contained in the following section.  

One of the early applications of the EM-algorithm is mixture models \citep{laird1978nonparametric}. Solving \eqref{KMsimp} with the EM-algorithm for mixture models with predetermined grid-points is especially simple, because the problem is convex (recall that the finite mixture model problem --- as opposed to the nonparametric problem --- is typically non-convex).  \cite{koenker2014convex} have developed interior point methods for solving \eqref{KMsimp}.  Along with \cite{koenker2016rebayes}, they created an R package \texttt{REBayes} that solves \eqref{KMsimp} for a handful of specific nonparametric mixture models, e.g. Gaussian mixtures and univariate Poisson mixtures; the \texttt{REBayes} package calls an external optimization software package, Mosek, and relies on Mosek's built-in interior point algorithms.  In our numerical analyses, we used \texttt{REBayes} to compute some one-dimensional NPMLEs with interior point methods.  To estimate multi-dimensional NPMLEs with interior point methods, we used our own implementations based on another R package \texttt{Rmosek} \citep{mosek}  (\texttt{REBayes} does appear to have some built-in functions for estimating two-dimensional NPMLEs, but we found them to be somewhat unstable in our applications).  We note that our interior point implementation solves the primal problem \eqref{KMsimp}, while \texttt{REBayes} solves the dual.  The Frank-Wolfe algorithm \citep{frank1956algorithm} is a classical algorithm for constrained convex optimization problems, which has recently been the subject of renewed attention \cite[e.g.][]{jaggi2013revisiting}.  Our implementation of the Frank-Wolfe algorithm closely resembles the ``vertex direction method,'' which has previously been used for finding the NPMLE in nonparametric mixture models \citep{bohning1995review}.  

All of the algorithms used in this paper were implemented in R.  We did not attempt to heavily optimize any of these implementations; instead, our main objective was to demonstrate that there are a range of simple and effective methods for finding (approximate) NPMLEs.  While the \texttt{REBayes} and \texttt{Rmosek} packages were used for their interior point methods, no packages beyond base R were required for any of our other implementations.  

\section{Simulation studies}

This section contains simulation results for NPMLEs and a Gaussian location-scale mixture model.  Section \ref{sec:numAlg} contains a comparison of the various NPMLE algorithms described in the previous section.  In Section \ref{sec:numMean}, we compare the performance of NPMLE-based estimators to other commonly used methods for estimating the mean in a Gaussian location-scale model.  

In all of the simulations described in this section we generated the data as follows.  For $j = 1,\ldots,p$, we generated independent $\Th_j = (\mu_j,\s_j) \sim G_0$ and corresponding observations $X_j \in \R^n$.  Each $X_j  = (X_{1j},\ldots,X_{nj})^{\T} $ was a vector of $n$ replicates
$X_{1j},\ldots,X_{nj}\mid\Th_j \sim N(\mu_j,\s_j^2)$ that were generated independently, conditional on $\Th_j$. In other words, 
\begin{equation}\label{GLS}
X_{j}\mid\Th_j \sim f_0(X_j|\Th_j)= \left(\frac{1}{\sqrt{2\pi}\sigma_j}\right)^n\exp\left\{-\frac{\sum_{i=1}^n(X_{ij}-\mu_j)^2}{2\sigma_j^2}\right\}.
\end{equation}
In the general model \eqref{GLS}, the mixing distribution $G_0$ is bivariate.  However, we considered two values of $G_0$ in our simulations: one where the marginal distribution of $\s_j$ was degenerate (i.e. $\s_j$ was constant, so $G_0$ was effectively univariate) and one  where the marginal distribution of $\s_j$ was non-degenerate.  (Note that for non-degenerate $\s_j$, $n \geq 2$ replicates are essential in order to ensure that the likelihood is bounded and that $\hat{G}$ exists.)

Throughout the simulations, we took $p = 1000$ and $n = 16$.  For the first mixing distribution $G_0$ (degenerate $\s_j$), we fixed $\s_j = 4$ and took $\mu_j$ so that $\P(\mu_j = 0) = \P(\mu_j = 5) = 1/2$.  For the second mixing distribution (non-degenerate $\s_j$), we took $\P(\mu_j = 0, \ \s_j = 5) = \P(\mu_j = 5, \ \s_j = 3) = 0.5$; for this distribution $\mu_j$ and $\s_j$ are correlated.

\subsection{Comparing NPMLE algorithms}\label{sec:numAlg}

For each mixing distribution, we computed $\hat{G}_{\L}$ using the algorithms described in Section \ref{sec:implementation}: the EM-algorithm, an interior point method with \texttt{Rmosek}, and the Frank-Wolfe algorithm.  For each of these algorithms, we also computed $\hat{G}_{\L}$ for various grids $\L$.  Specifically, we considered regular grids $\L = \{m_k\}_{k= 1}^{q_1} \times \{s_k\}_{k = 1}^{q_2}  \subseteq [\min \hat{\mu}_j, \max\hat{\mu_j}] \times [\min \hat{\s}_j, \max\hat{\s}_j]  \subseteq \R^2$, where $\hat{\mu}_j = n^{-1}\sum_i X_{ij}$ and $\hat{\s}_j^2 = n^{-1}\sum_i (X_{ij} - \hat{\mu}_j)^2$.  The values $q_1,q_2$ determine number of grid-points in $\L$ for $\mu_j,\s_j$, respectively, and in the simulations we fit estimators with $(q_1,q_2) = (30,30), (50,50)$, and $(100,100)$. The stopping criterion of EM and Frank-Wolfe is
		\bes
		|\ell(w^{(K)})-\ell(w^{K-1})|/|\ell(w^{(K-1)}| \le 10^{-6},
		\ees 
where $\ell(w)$ is the objective function in (\ref{KMsimp}) and $w^{(K)}$ is the solution at the $K$-th iteration of respective methods.

In addition to fitting the two-dimensional NPMLEs $\hat{G}_{\L}$ described above, for the simulations with degenerate $\s_j$ we also fit one-dimensional NPMLEs to the data $\hat{\mu}_1,\ldots,\hat{\mu}_p$, according to the model
\[
\hat{\mu}_j \mid \mu_j \sim N(\mu_j,1).
\]
We fit these one-dimensional NPMLEs using all of the same algorithms for the two-dimensional NPMLEs (EM, interior point with \texttt{Rmosek}, and Frank-Wolfe), and we also used the \texttt{REBayes} interior point implementation to estimate the distribution of $\mu_j$ in this setting.  For the one-dimensional NPMLEs, we took $\L \subseteq [\min \hat{\mu}_j,\max\hat{\mu}_j] \subseteq \R$ to be the regular grid with $q = 300$ points.  This allows us to compare the performance of methods for one- and two-dimensional NPLMEs (where the one-dimensional NPMLEs take the distribution of $\s_j$ to be known) and compare the performance of the two interior point algorithms, among other things.

For each simulated dataset and estimator $\hat{G}_{\L}$ we recorded several metrics.  First, we computed the total squared error (TSE)
	\bes
	\TSE = \sum_{j = 1}^p (\muhat_j - \mu_j)^2,
	\ees
	where
	\bes
	\hat{\mu}_j = \E_{\hat{G}_{\L}}(\mu_j \mid X_j)=\frac{\int\mu_jf_0(X_j|\mu_j,\sigma_j)d\hat{G}_{\L}(\mu_j,\sigma_j)}{\int f_0(X_j|\mu_j,\sigma_j)d\hat{G}_{\L}(\mu_j,\sigma_j)}. 
	\ees
Second, we computed the difference between the log-likelihood of $\hat{G}_{\L}$ and the log-likelihood of $\hat{G}_{\EM}$, the corresponding estimator for $G_0$ based on the EM-algorithm:
\[
\Delta(\mbox{log-lik.}) = \ell(\hat{G}_{\EM}) - \ell(\hat{G}_{\L}).
\]
Note that $\Delta(\mbox{log-lik.}) > 0$ if $\hat{G}_{\L}$ has a {\em smaller} negative log-likelihood than $\hat{G}_{\EM}$ (we are taking the EM-estimator $\hat{G}_{\EM}$ as a baseline for measuring the log-likelihood). Finally, we recorded the time required to compute $\hat{G}_{\L}$ (in seconds; all calculations were performed on a 2015 MacBook Pro laptop).  Summary statistics are reported in Table \ref{tab:alg}.

It is evident that the results in Table \ref{tab:alg} are relatively insensitive to the number of grid-points $(q_1,q_2)$ chosen for the two-dimensional NPMLE implementations.  In terms of TSE, the EM algorithm and interior point methods perform very similarly across all of the settings, while the interior point methods appear to slightly out-perform the EM algorithm in terms of $\Delta(\mbox{log-lik.})$ across the board.  Additionally, the interior point methods have smaller compute time than the EM algorithm, though the difference is not too significant for applications at this scale (for mixing distribution 1, with degenerate $\s_j$, the \texttt{REBayes} dual implementation appears to be somewhat faster than our \texttt{Rmosek} primal implementation).  The Frank-Wolfe algorithm is the fastest implementation we have considered, but its performance in terms of TSE and $\Delta(\mbox{log-lik.})$ is considerably worse than the EM algorithm or interior point methods.  In the remainder of the paper, we chose to use the EM algorithm exclusively for computing NPMLEs --- we believe it strikes a balance between simplicity and performance.  

\begin{table}[H]
	\begin{center}
		\caption{Comparison of different NPMLE algorithms.  Mean values (standard deviation in parentheses) reported from 100 independent datasets; $p=1000$, throughout simulations. Mixing distribution 1 has constant $\s_j$; mixing distribution 2 has correlated $\mu_j$ and $\s_j$.}
		\label{tab:alg}
		\setlength{\tabcolsep}{1pt}
		\begin{tabular}{clccc}
			&& TSE &\begin{tabular}[b]{c} $\Delta(\text{log-lik})$\\ $\! \times\! 10^{4}$\end{tabular} &\begin{tabular}[b]{c} Time\\ (secs.) \end{tabular}\\ \hline
			Mixing dist. 1&EM &&& \\
			(Bivariate)&\ \ \ $(q_1,q_2) = (30,30)$ &130.5  (42.6) & 0 (0)& 9\\ 
			&\ \ \  $(q_1,q_2) = (50,50)$& 130.4 (42.7)& 0 (0)& 33\\
			&\ \ \  $(q_1,q_2) = (100,100)$&130.4 (42.6) &0 (0)& 136
			\\ 
			&Interior point (\texttt{Rmosek}) &&&\\ 
			&\ \ \ $(q_1,q_2) = (30,30)$&130.7 (42.6)& 6 (1)& 8\\
			&\ \ \ $(q_1,q_2) = (50,50)$& 130.5 (42.9)& 9 (1)&20 \\ 
			&\ \  \ $(q_1,q_2) = (100,100)$&130.6 (42.8) & 11 (1)  &80 \\ 
			&Frank-Wolfe  &&& \\ 
			&\ \ \ $(q_1,q_2) = 
			(30,30)$&147.3 (45.1)& -234 (130)& 5 \\
			&\ \ \   $(q_1,q_2) = (50,50)$& 147.0 (45.9)& -238 (134)& 14\\
			&\ \ \  $(q_1,q_2) = (100,100)$&146.2 (45.4)& -238 (128)& 55\\\hline
			Mixing dist. 1&EM&124.4 (41.6) & 0 (0)& 1
			\\
			(univariate; $q = 300$;	&Interior point (\texttt{Rmosek}) &124.3 (42.1)& 6 (1)&3\\
			assume known $\sigma_j$)&Interior point (\texttt{REBayes})&124.3 (42.1)& 6 (1) &1  \\  
			&	Frank-Wolfe & 126.1 (41.8)& -4 (4)& 1 \\
			\hline
			Mixing dist. 2&		EM &&& \\
			&	\ \ \ $(q_1,q_2) = (30,30)$& 54.0 (28.4) &0 (0)& 9\\ 
			&\ \ \  $(q_1,q_2) = (50,50)$&  54.0 (28.9) & 0 (0) &34\\
			&		\ \ \  $(q_1,q_2) = (100,100)$&53.9 (28.8) &0 (0) &141 \\ 
			&	Interior point (\texttt{Rmosek}) &&& \\ 
			&	\ \ \ $(q_1,q_2) = (30,30)$& 54.3 (28.8) & 5 (1) &8\\  
			&	\ \ \ $(q_1,q_2) = (50,50)$&54.2 (29.1)& 8 (1)&20\\ 
			&	\ \  \ $(q_1,q_2) = (100,100)$ &54.3 (29.1) & 10 (1) &82 \\
			&	Frank-Wolfe  &&& \\ 
			&	\ \ \ $(q_1,q_2) = 
			(30,30)$& 82.2 (39.3)& -372 (217) &5 \\
			&	\ \ \   $(q_1,q_2) = (50,50)$&83.1 (36.4) & -402 (232) &14\\
			&	\ \ \  $(q_1,q_2) = (100,100)$&82.0 (37.7) & -396 (240) & 56\\
			\hline
		\end{tabular}
	\end{center}
\end{table}

\subsection{Gaussian location-scale mixtures: Other methods for estimating a normal mean vector} \label{sec:numMean}

Beyond NPMLEs, we also implemented several other methods that are commonly used for estimating the mean vector $\boldsymbol{\mu} = (\mu_1,\ldots,\mu_p)^\T \in \R^p$ in Gaussian location-scale models and computed the corresponding TSE.  Specifically, we considered the fixed-$\Th_j$ MLE, $\boldsymbol{\muhat} = (\hat{\mu}_1,\ldots,\hat{\mu}_p)^\T \in \R^p$; the James-Stein estimator; the heteroscedastic SURE estimator of \citet{xie2012sure}; and a soft-thresholding estimator. The James-Stein estimator is a classical shrinkage estimator for the
Gaussian location model.  The version employed here is described in \cite{xie2012sure} and is designed for
heteroscedastic data. The heteroscedastic SURE estimator is another shrinkage estimator,
which was designed to ensure that features with a high noise variance
are ``shrunk'' more than those with a low noise variance. Both the James-Stein estimator and the heteroscedastic SURE estimator depend on the values $\s_j$.  
The soft-thresholding
estimator takes the form $\hat{\mu}(X) = s_t(\hat{\mu}_j)$, where $t \geq  0$ is a constant and $s_t(x) = \mathrm{sign}(x)\max\{|x| - t,0\}$,
$x \in \R$.  For soft-thresholding estimators, $t$ was chosen to minimize the TSE.  Observe that the James-Stein, SURE, and soft-thresholding estimators all depend on information that is typically not available in practice: the value of $\s_j$ and the actual TSE.  By contrast, the two-dimensional NPMLEs described in the previous sub-section utilize only the observed data $X_1,\ldots,X_p$.

In Table \ref{tab:est}, we report the TSE for the different estimators described in this section, along with the TSE for the bivariate NPMLE fit using the EM algorithm.  We also fit a univariate NPMLE in this example, where $\s_j$ was {\em not} assumed to be known; instead we used the plug-in estimator $\hat{\s}_j$ in place of $\s_j$ and then computed the NPMLE for the distribution of $\mu_j$.   

Table \ref{tab:est} shows that the NPMLEs dramatically out-perform the alternative estimators in this setting, in terms of TSE. The bivariate NPMLE out-performs the univariate NPMLE under both mixing distributions 1 and 2, but its advantage is especially pronounced under mixing distribution 2, where $\mu_j$ and $\s_j$ are correlated.  This highlights the potential advantages of bivariate NPMLEs over univariate approaches in settings with multiple parameters.  

\begin{table}[h]
	\begin{center}
		\caption{Mean TSE for various estimators of $\boldsymbol{\mu} \in \R^p$ based on 100 simulated datasets; $p = 1000$. $(q_1,q_2)$ indicates the grid-points used to fit $\hat{G}_{\L}$.}
		\label{tab:est}
		\setlength{\tabcolsep}{1.2pt}
		\begin{tabular}{lccclccc}
			Method & Mixing dist.1 & \ Mixing dist.2 \\ \hline 
			Fixed-$\Th_j$ MLE & 997.0 (48.2)&1059.3 (57.9)\\
			Soft-Thresholding &826.2 (50.0)& 793.7 (46.0) \\
			James-Stein & 859.7 (43.6)&935.2 (53.0)\\
			SURE &859.7 (43.6)& 880.7 (48.1)\\
			Univariate NPMLE $q = 300$ & 170.7 (47.0) & 285.4 (63.5)\\ 
			Bivariate NPMLE $(q_1,q_2)=(100, 100)$ &130.4 (42.6)&53.9 (28.8)
		\end{tabular}		
	\end{center}
\end{table}

\section{Baseball data}
\label{sec:bball}

Baseball data is a well-established testing
ground for empirical Bayes methods \citep{efron1975data}.  The
baseball dataset we analyzed contains the number of at-bats and
hits for all of the Major League Baseball players during the 2005
season and has been previously analyzed in a number of papers \citep{brown2008season,jiang2010empirical,muralidharan2010empirical,xie2012sure}.
The goal of the analysis is to use the data from the first half of the
season to predict each player's batting average (hits/at-bats) during
the second half of the season.  Overall, there are 929 players in the
baseball dataset; however, following \cite{brown2008season} and others, we
restrict attention to the 567 players with more than 10
at-bats during the first half of the season for training and among which 499 players with more than 10 at-bats during the second half of season for prediction. (We follow the other
preprocessing steps described in \cite{brown2008season} as well).

Let $A_{j}$ and $H_{j}$ denote the number
of at-bats and hits, respectively, for player $j$ during the first
half of the season, $j=1,...,567$.  We assume that $(A_j,H_j)$ follows a
Poisson-binomial mixture model, where $A_j\mid (\l_j,\pi_j) \sim \mbox{Poisson}(\l_j)$,  $H_j\mid
(A_j,\l_j,\pi_j) \sim \mbox{binomial}(A_j,\pi_j)$,
and $(\l_j,\pi_j) \sim G_0$.  This model has a bivariate mixing distribution $G_0$, i.e. $d = 2$.  In the notation of \eqref{MM}, $X_j =
(A_j,H_j)$, $\Th_j = (\l_j,\pi_j)$ and
\bes
f_0(X_j|\Th_j)=f_0(A_j, H_j|\lambda_j,\pi_j)=\frac{\lambda_j^{A_j}e^{-\lambda_j}}{A_j!}\times {{A_j}\choose{H_j}}\pi_j^{H_j}(1-\pi_j)^{A_j-H_j}.
\ees
We took $\L\subseteq [\min A_j, \max A_j]\times [ \min A_j/H_j, \max A_j/H_j]$ to be the regular grid with $(q_1,q_2)=(30,30)$.
We propose
to estimate each
player's batting average for the second half of the season by
the posterior mean of $\pi$, computed under $(\l,\pi) \sim
\hat{G}_{\L}$, 
\begin{equation}\label{pi_hat}
\hat{\pi}_j = \E_{\hat{G}_{\L}}(\pi_j \mid A_j,H_j)=\frac{\int\pi_jf_0(A_j, H_j|\lambda_j,\pi_j)d\hat{G}_{\L}(\lambda_j,\pi_j)}{\int f_0(A_j, H_j|\lambda_j,\pi_j)d\hat{G}_{\L}(\lambda_j,\pi_j)}.  
\end{equation}

Most previously published analyses of the baseball data begin by transforming
the data via the variance stabilizing transformation 
\begin{equation}\label{stab}
W_j = \arcsin 
\sqrt{\frac{H_j + 1/4}{A_j + 1/2}}
\end{equation}
(\citet{muralidharan2010empirical} is a notable exception).
Under this transformation,
$W_j$ is approximately distributed as $N\{\mu_j,(4A_j)^{-1}\}$, where $\mu_j = \arcsin
\sqrt{\pi_j}$. Methods for Gaussian observations may be applied to the transformed
data, with the objective of estimating $\mu_j$.  Following this approach, a variety of methods based on shrinkage,
the James-Stein estimator, and parametric empirical
Bayes methods for Gaussian data
have been
proposed and studied \citep{brown2008season,jiang2010empirical,xie2012sure}.

Under the transformation \eqref{stab}, it is standard to use total squared error to measure
the performance of estimators $\hat{\mu}_j$ \citep[e.g.][]{brown2008season}.  In this example, the total squared error is defined as
\[
\mathrm{TSE} = \sum_{j=1,...,499} \left\{(\hat{\mu}_j - \tilde{W}_j)^2 -
\frac{1}{4\tilde{A}_j}\right\},
\]
where 
\[
\tilde{W}_j = \arcsin \sqrt{\frac{\tilde{H}_j +
		1/4}{\tilde{A}_j + 1/2}},
\] and $\tilde{A}_j$ and $\tilde{H}_j$ denote the
at-bats and hits from the second half of the season, respectively.
For convenience of comparison, we used $\mathrm{TSE}$ to measure
the performance of our estimates $\hat{\pi}_j$, after applying the
transformation $\hat{\mu}_j = \arcsin \sqrt{\hat{\pi}_j}$.

Results from the baseball analysis are reported in
Table \ref{tab:bball}.  Following the work of others, we have analyzed
all players from the dataset together, and then the pitchers and
non-pitchers from the dataset separately. 
 In addition to our
Poisson-binomial NPMLE-based estimators \eqref{pi_hat}, we
considered six other previously studied estimators:
\begin{enumerate}
	\item  The (fixed-parameter) {\em MLE}
	estimator $\hat{\mu}_j = W_j$ uses each player's hits and at-bats from the first half of
	the season to estimate their performance in the second half.
	\item The
	{\em grand mean} $\hat{\mu}_j = p^{-1}(W_1 + \cdots + W_p)$ gives the
	exact same estimate for each player's performance in the second half of the season,
	which is equal to the average performance of all players during the
	first half.
	\item The {\em James-Stein} parametric empirical Bayes estimator described in
	\cite{brown2008season}.
	\item The weighted generalized MLE ({\em weighted GMLE}), which uses at-bats as a
	covariate \citep{jiang2010empirical}.  This is essentially a
	univariate NPMLE-method for Gaussian models with covariates.
	\item The {\em semiparametric SURE} estimator is a flexible shrinkage estimator that may be viewed as a
	generalization of the James-Stein estimator \citep{xie2012sure}.  
	\item The {\em binomial mixture} method in
	\cite{muralidharan2010empirical} is another empirical Bayes method, which does not require the data to be
	transformed and estimates $\pi_j$ directly (in \cite{muralidharan2010empirical}, they work conditionally on
	the at-bats $A_j$).  $\mathrm{TSE}$ is computed after applying the $\mathrm{arcsin}\sqrt{\cdot}$ transformation.
\end{enumerate}

\begin{table}[h]
	\begin{center}
		\caption{Baseball data.  $\mathrm{TSE}$ relative to the
			MLE. Minimum error is in {\bf bold} for each analysis.}
		\label{tab:bball}
		\setlength{\tabcolsep}{1pt}
		\begin{tabular}{lccclccc}
			& &Non- & &&& Non- & \\ Method &All&Pitchers&
			Pitchers
			& Method & All & Pitchers & Pitchers  \\ \hline 
			 MLE& 1&1&1 & SURE  & 0.41 &
			{\bf 0.26}
			& {\bf 0.08} \\ 
			Grand mean & 0.85& 0.38& 0.13 & Binomial
			mixture & 0.59
			& 0.31 & 0.16 \\ 
			James-Stein &0.54& 0.35&0.17 &
			NPMLE
			& {\bf
				0.29} &
			{\bf 0.26} & 0.14 \\ 
			GMLE &0.30& {\bf 0.26}&0.14   
		\end{tabular}		
	\end{center}
\end{table}

The values reported in Table \ref{tab:bball} are the $\mathrm{TSE}$s of each
estimator, 
relative to the $\mathrm{TSE}$ of the fixed-parameter MLE.  Our
Poisson-binomial method performs very well, recording the minimum
$\mathrm{TSE}$ when all of the data (pitchers and non-pitchers) are
analyzed together and
for the non-pitchers.  Moreover, the Poisson-binomial NPMLE
$\hat{G}_{\L}$ works on the original scale of the data (no
angular transformation is required) and may be
useful for other purposes, beyond just estimation/prediction. Figure
\ref{fig:bball_dist} (a) plots the marginal Poisson-binomial NPMLE distribution of $\pi$, estimated with EM algorithm using all players, non-pitchers and pitchers data. Figure
\ref{fig:bball_dist} (b) contains a histogram and contours of 20,000 independent draws from the
estimated mixture distribution of $(A_j,H_j/A_j)$, fitted with the
Poisson-binomial NPMLE to all players in the baseball dataset.  Observe that
the distribution appears to be bimodal.   By comparing
this histogram with histograms of the observed data from the
non-pitchers and pitchers separately (Figure \ref{fig:bball_dist} (c)--(d)),
it appears that the mode at the left of Figure \ref{fig:bball_dist} (b)
represents a group of players that includes the pitchers and the mode at the right
represents the bulk of the non-pitchers.

\begin{figure*}[t]
	\centering
	\includegraphics[width=\textwidth]{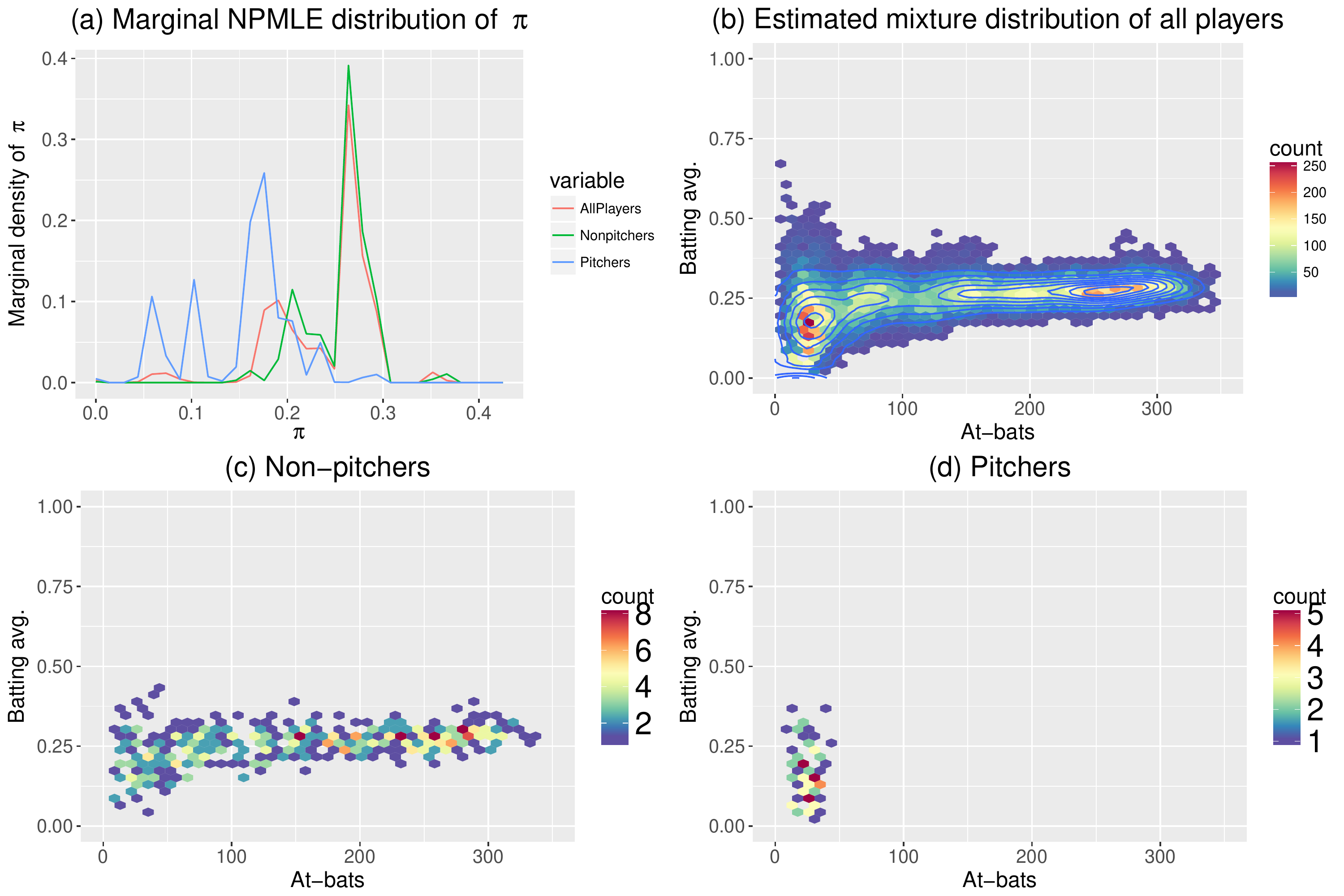}
	\caption{\label{fig:bball_dist} (a) Marginal NPMLE distribution of $\pi$, computed using EM algorithm; (b) histogram of 20,000 independent
		draws from the estimated distribution of $(A_j,H_j/A_j)$, fitted with the
		Poisson-binomial NPMLE to all players in the baseball dataset; (c)
		histogram of non-pitcher
		data from the baseball dataset; (d) histogram of pitcher data from the baseball
		dataset.}
\end{figure*}

\section{Two-dimensional NPMLE for cancer microarray classification}
\label{sec:class}

\cite{dicker2016nonparametric} proposed a univariate NPMLE-based method for
high-dimensional classification problems and studied  applications
involving cancer microarray data.  The classifiers from
\cite{dicker2016nonparametric} are based on a Gaussian model
with one-dimensional mixing
distributions, i.e. $d=1$.  In this section we show that using a bivariate
mixing distribution may substantially improve performance.  

Two datasets from the Microarray
Quality Control Phase II project \citep{maqc2010microarray} are considered; one
from a breast cancer study and one from a myeloma study.  The
training dataset for the breast cancer study contains $n = 130$
subjects and $p = 22283$ probesets (genes); the test dataset
contains $100$ subjects.  The training dataset for
the myeloma study contains $n = 340$ subjects and $p = 54675$
probesets; the test dataset contains $214$ subjects.   The goal
is to use the training data to build binary classifiers for several
outcomes, then check the
performance of these classifiers on the test data.   Outcomes for
the breast cancer data are response to treatment (``Response'') and
estrogen receptor status (``ER status''); outcomes for the myeloma
data are overall and event-free survival (``OS'' and ``EFS'').  

For each
of the studies, let
$X_{ij}$ denote the expression level of gene $j$ in subject $i$ and
let $Y_i \in \{0,1\}$ be the class label for subject $i$.  Let $X_j =
(X_{1j},\ldots,X_{nj}) ^{\T} \in \R^n$.  We assume that each class ($k
= 0,1$) and
each gene ($j=1,\ldots,p$) has an associated mean expression level
$\mu_{jk} \in \R$, and that conditional on the $Y_i$ and $\mu_{jk}$ all of
the $X_{ij}$ are independent and Gaussian, satisfying $X_{ij}\mid (Y_i = k, \ \mu_{jk}) \sim
N(\mu_{jk},1)$  (the gene-expression levels in the datasets are all
standardized to have variance 1).  

In \cite{dicker2016nonparametric}, they assume
that $\mu_{1k},\ldots,\mu_{pk} \sim G_k$ ($k = 0,1$) are all independent draws
from two distributions, $G_0$ and $G_1$.  They use the
training data from classes $k= 0,1$ to separately estimate the
distributions $G_0$ and $G_1$ using NPMLEs, and then implement the Bayes
classifier, replacing $G_0$ and $G_1$ with the corresponding
estimates.   In this paper, we model $\Th_j = (\mu_{j0},\mu_{j1})
\sim G_0$ jointly, then compute the bivariate
NPMLE $\hat{G}_{\L}$, and finally use $\hat{G}_{\L}$ in place of $G_0$
in the Bayes classifier for this model.   
In the notation of \eqref{MM}, $\Th_j = (\mu_{j0},\mu_{j1})$ and
	\bes
	X_j\mid\Th_j \sim f_0(X_j|\Th_j)\propto\exp\left[-\frac{1}{2}\sum_{i=1}^n\left\{X_{ij}-\sum_{k=0,1}\mu_{jk}\mathbbm{1}(Y_i=k)\right\}^2\right],
	\ees
where $\mathbbm{1}(\cdot)$ is the indicator function. The model from
\cite{dicker2016nonparametric} is equivalent to the model proposed
here, when $\mu_{j0}$ and $\mu_{j1}$ are independent.   Results
from analyzing the MAQC datasets using these two
classifiers (the previously proposed method with 1-dimensional
NPMLEs and the 2-dimensional NPMLE described here), along with some
other well-known and relevant classifiers, may be found in
Table \ref{tab:class}. The other classifiers we considered were:

\begin{enumerate}
	\item {\em NP EBayes w/smoothing}.  Another nonparametric empirical
	Bayes classifier proposed in \cite{greenshtein2009application}, which
	uses nonparametric smoothing to fit a univariate density to the
	$\mu_{jk}$ 
	and then employs a version of linear discriminant analysis. 
	\item {\em  Regularized LDA}. A version of $\ell^1$-regularized linear discriminant analysis, proposed in
	\citet{mai2012direct}. 
	\item {\em Logistic lasso}. $\ell^1$-penalized logistic regression fit
	using the R package \texttt{glmnet}.  
\end{enumerate}

\begin{table}[t]
	\begin{center}
		\caption{Microarray data.  Number of misclassification
			errors on test data. }
		\label{tab:class}
		\setlength{\tabcolsep}{3pt}
		\begin{tabular}{llccccccc} &&&&&& &Logistic & \\
			Dataset&
			Outcome& $n_{\mathrm{test}}$
			&2d-NPMLE&1d-NPMLE& EBayes
			& LDA
			&lasso
			\\ \hline
			Breast&Response& 100 &{\bf 15}&36&47&30&18\\
			Breast&ER status& 100 &19&40&39&{\bf 11} & {\bf 11}\\
			Myeloma&OS& 214 &30&55&100&97&{\bf 27}\\
			Myeloma&EFS& 214&34&76&100&63&{\bf 32}
		\end{tabular}		
	\end{center}
\end{table}

For each of the datasets and outcomes, the 2-dimensional NPMLE
classifier substantially outperforms the 1-dimensional NPMLE, and is very
competitive with the top performing classifiers.  Modeling dependence
between $\mu_{j0}$ and $\mu_{j1}$, as with the 2-dimensional NPMLE,
seems sensible because most of the genes are likely to have similar
expression levels across classes, i.e. $\mu_{j0}$ and $\mu_{j1}$ are
likely to be correlated.  This may be interpreted as a kind of sparsity
assumption on the data, which is prevalent in high-dimensional classification
problems.  Moreover, our proposed method involving NPMLEs
should adapt to non-sparse settings as well, since $G_0$ is allowed to
be an arbitrary bivariate distribution.  

One of the main underlying assumptions of the NPMLE-based
classification methods is that the different genes have independent
expression levels.  This is certainly {\em not} true in most applications,
but is similar in principle to a ``naive Bayes'' assumption.
Developing methods for NPMLE-based classifiers to better handle correlation
in the data may be of interest for future research.  

\section{Continuous glucose monitoring}
\label{sec:cgm}

The analysis in this section is based on blood glucose data from a
study involving 137 type 1 diabetes patients; more details on the study may be
found in \cite{hirsch2008sensor} and \cite{dicker2013continuous}.  Subjects in
the study were monitored for an average of 6 months each.  Throughout
the course of the
study, each subject wore a continuous glucose monitoring device, built around
an electrochemical glucose biosensor.  Every 5 minutes while in use,
the device records (i) a raw electrical current measurement from the
sensor (denoted $\ISIG$), which is
known to be correlated with blood glucose density, and (ii) a timestamped
estimate of blood glucose density ($\CGM$), which is based on a proprietary algorithm
for converting the available data (including the electrical current
measurements from the sensor) into blood glucose density estimates.
In addition to using the sensors, each study subject maintained a
self-monitoring routine, whereby blood glucose density was measured
approximately 4 times per day from a small blood sample extracted by
fingerstick.   Fingerstick measurements of blood glucose density are considered to be more
accurate (and are more invasive) than the sensor-based estimates
(e.g. CGM).
During the study, the
result of each fingerstick measurement was manually entered into the
continuous monitoring device at the time of measurement; algorithms
for deriving continuous sensor-based estimates of blood glucose
density, such as $\CGM$, may use the available fingerstick
measurements for calibration purposes.

In the rest of this section, we show how NPMLE-based empirical
Bayes methods can be used to improve algorithms for estimating blood glucose
density using the continuous monitoring data.  The basic idea is that
after formulating a statistical model relating blood glucose density
to $\ISIG$, we allow for the possibility that the model parameters may
differ for each subject, then use a training dataset to estimate the
distribution of model parameters across subjects (i.e. estimate
$G_0$) via nonparametric maximum likelihood.  This is illustrated for
two different statistical models in Sections \ref{sec:LM}--\ref{sec:KF}.

Throughout the analysis below, we use fingerstick measurements as a proxy
for the actual blood glucose density values.  Let $\FS_j(t)$ and
$\ISIG_j(t)$ denote the fingerstick blood glucose density and $\ISIG$
values, respectively, for the $j$-th subject at time $t$.   Recall
that $\FS_j(t)$ is measured, on average, once every 6
hours, while $\ISIG_j(t)$ is available every 5 minutes.  Let
$\F_t$ denote the $\s$-field of information available at time $t$
(i.e. all of the fingerstick and $\ISIG$ measurements taken before
time $t$, plus $\ISIG_j(t)$).  For each methodology, we use the first half of the available
data for each subject to fit a statistical model relating $\ISIG_j(t)$
to $\FS_j(t)$, then estimate each value $\FS_j(t)$ in the second half
of the data using $\widehat{\FS}_j(t)$, an estimator based on $\F_t$.
The performance of each method is measured by the MSE
on the test data, relative to the MSE of the proprietary estimator
$\CGM$.  

\subsection{Linear model}\label{sec:LM}

First we consider a basic linear regression model relating $\FS$ and $\ISIG$, 
\begin{equation}\label{LM}
\FS_j(t) = \mu_j + \beta_j \ISIG_j(t) + \s_j \epsilon_j(t),  \ \ t=t_1,...,t_n
\end{equation}
where the $\epsilon_j(t)$ are iid $N(0,1)$, $t_1,...,t_n$ is the FS measurement time and $\Th_j
= \{\mu_j,\beta_j,\log(\s_j)\} \in \R^3$ are unknown, subject-specific
parameters. In the notation of \eqref{MM}, $X_j =\FS_j\equiv\{\FS_j(t), t=t_1,..,t_n\}$ and
	\bes
	X_j\mid\Th_j &\sim&  f_0(\FS_j|\Th_j)\cr&\propto&\exp\left[-n\log\sigma_j-\frac{\sum_{i=1}^n\big\{\FS_j(t_i)-\mu_j- \beta_j \ISIG_j(t_i)\big\}^2}{2\exp(2\log\sigma_j)}\right].
	\ees
Three ways to fit \eqref{LM} are (i) using the combined model,
where $\Th_j = \Th = \{\mu,\beta,\log(\s)\}$ for all $j$, i.e. all the
subject-specific parameters are the same; (ii) the individual model,
where $\Th_1,\ldots,\Th_p$ are all estimated separately, from the
corresponding subject data; and (iii) the nonparametric mixture model, where $\Th_j
\sim G_0$ are iid draws from the $d=3$-dimensional mixing distribution
$G_0$. 
 For each of these methods, we took $\widehat{\FS}_j(t) =
\hat{\mu}_j + \hat{\beta}_j \ISIG(t)$, where $\hat{\mu}_j$
and $\hat{\beta}_j$ are the corresponding MLEs under the combined and
individual models, and, under the mixture model, $\hat{\mu}_j =
\E_{\hat{G}_{\L}}(\mu_j\mid \F_t)$ and $\hat{\beta}_j =
\E_{\hat{G}_{\L}}(\beta_j\mid \F_t)$.  Results are reported in Table
\ref{tab:BG}.

\begin{table}[t]
	\begin{center}
		\caption{Blood glucose data. MSE relative to CGM.}\label{tab:BG}
		\begin{tabular}{cccccc}
			\multicolumn{3}{c}{Linear model} &
			\multicolumn{3}{c}{Kalman filter} \\ 
			Combined & Individual & NPMLE & Combined & Individual & NPMLE
			\\ \hline 1.56 & 1.54 & {\bf 1.51} & 1.05 & 1.07 &
			{\bf 1.03}
		\end{tabular}
	\end{center}
\end{table}

\subsection{Kalman filter}\label{sec:KF}

Substantial performance improvements are possible by allowing the
model parameters relating $\FS$ and $\ISIG$ to vary with time.  In
this section we consider the Gaussian state space model (Kalman filter) 
\begin{equation}\label{KF}
\begin{array}{rcl}
\FS_j(t_i) \!\!\!\!& = &\!\!\!\! \alpha_j(t_i)\ISIG_j(t_i) + \s_j\e_j(t_{i-1}), \\
\alpha_j(t_i) \!\!\!\!& = & \!\!\!\!\alpha_j(t_{i-1}) + \tau_j
\delta_j(t_{i-1}), 
\end{array}
\end{equation}
where we assume that $\FS$ along with $\ISIG$ are observed at times $t_1,\ldots,t_n$ and $\e_j(t),\d_j(t) \sim N(0,1)$ are iid.  In \eqref{KF},
$\{\alpha_j(t)\}$ are the state variables that evolve according to a
random walk and $\Th_j = \{\log(\tau_j),\log(\s_j)\}$ are unknown
parameters.  Unlike \eqref{LM}, there is no intercept term in
\eqref{KF}; dropping the intercept term has been previously justified when using state
space models to analyze glucose sensor data
\citep[e.g.][]{dicker2013continuous}.  The parameters $\s_j,\tau_j$
control how heavily recent observations are weighted when estimating
$\alpha_j(t)$.

Similar to the analysis in Section \ref{sec:LM}, we fit \eqref{KF}
using (i) a combined model where $\Th_j = \Th$ for all $j$; (ii) an individual model where
$\Th_1,\ldots,\Th_p$ are estimated separately; and (iii) a
nonparametric mixture model,  where $\Th_j
\sim G_0$ are iid draws from a $d=2$-dimensional mixing distribution.
Under (i)--(ii), $\s_j$ and $\tau_j$ are estimated by maximum likelihood and $\widehat{\FS}_j(t_i) =
\E\{\alpha_j(t_i) \mid \F_{t_i}\}\times \ISIG_j(t_i)$, where the conditional
expectation is computed with respect to the Gaussian law governed by
\eqref{KF}, with $\hat{\s}_j$ and $\hat{\tau}_j$ replacing $\s_j$ and
$\tau_j$ (i.e. we use the Kalman filter).  For the nonparametric
mixture (iii), $\widehat{\FS}_j(t_i) =
\E_{\hat{G}_{\L}}\{\alpha_j(t_i) \mid \F_{t_i}\}\times \ISIG_j(t_i)$, where
the expectation is computed with respect to the model \eqref{KF} and the
estimated mixing
distribution $\hat{G}_{\L}$.   Results are reported in Table \ref{tab:BG}.

\subsection{Comments on results}

From Table \ref{tab:BG}, it is evident that the NPMLE mixture approach outperforms the
individual and combined methods for both the linear model and the
Kalman filter/state space model.  The Kalman filter methods perform substantially better than the
linear model, highlighting the importance of time-varying parameters
(scientifically, this is justified because the sensitivity of the
glucose sensor is known to change over time).  Note that all of the relative MSE values in
Table \ref{tab:BG} are greater than 1, indicating that $\CGM$ still outperforms
all of the methods considered here.  Somewhat more ad hoc methods for
estimating blood glucose density that do outperform CGM are
described in \cite{dicker2013continuous}; these methods (and CGM)
leverage additional data available to the continuous monitoring
system, which is not described here for the sake of simplicity.  The
methods in \cite{dicker2013continuous} are somewhat  similar to the ``combined''
Kalman filtering method from Section \ref{sec:KF}, where $\Th_j = \Th$ for all $j$; it would
be interesting to see if the performance of these methods could be further
improved by using NPMLE ideas.   

\section{Conclusion}
\label{sec:discussion}

We have proposed a flexible, practical approach to 
fitting general multivariate mixing
distributions with NPMLEs and illustrated the effectiveness of this
approach through several real data examples. Theoretically, we proved
that the support set of the NPMLE is a subset of the convex hull of MLEs when the likelihood $f_0$ comes from a class of
elliptical unimodal distributions.  We believe that this approach may
be attractive for many problems where mixture models and
empirical Bayes methods are relevant, offering both effective
performance and computational simplicity.  

\section*{Acknowledgments}
The authors are extremely grateful to the referee for their helpful and stimulating comments, leading to an improved exposition. The work of Lee H. Dicker is partially supported by NSF Grants DMS-1208785 and DMS-1454817. The work of Long Feng is supported by National Institute on Drug Abuse Grant R01 DA016750.



\bibliographystyle{elsarticle-harv} 
\bibliography{npmleref}





\end{document}